 \documentclass[a4paper,UKenglish,cleveref, autoref]{lipics-v2019}
 \usepackage[utf8]{inputenc}
 \usepackage{amsfonts}
 \usepackage{amsthm,amsmath,amssymb,mathtools}
 \usepackage{algorithm}
 \usepackage[noend]{algpseudocode}
 \usepackage{xcolor}

 \newcommand{\Ceil}[1]{\left\lceil{#1}\right\rceil}
 \newcommand{\Abs}[1]{\left|{#1}\right|}

 \newcommand{\etal}{{et al.}}

 \def\N{{\cal N}}

 \theoremstyle{definition}
 \newtheorem{observation}[theorem]{Observation}

 \title{Local Routing in Sparse and Lightweight Geometric Graphs} 
 \titlerunning{Routing in Sparse Graphs}

 \author{Vikrant Ashvinkumar}{University of Sydney, Australia}{vash7242@uni.sydney.edu.au}{}{}
 \author{Joachim Gudmundsson}{University of Sydney, Australia}{joachim.gudmundsson@sydney.edu.au}
 {}{Funded by the Australian Government through the Australian Research Council DP150101134 and DP180102870.}
 \author{Christos Levcopoulos}{Lund University, Sweden}{christos.levcopoulos@cs.lth.se}{}{Swedish Research Council grants 2017-03750 and 2018-04001.}
 \author{Bengt J. Nilsson}{Malmö University, Sweden}{bengt.nilsson.ts@mau.se}{}{}
 \author{Andr\'e van Renssen}{University of Sydney, Australia}{andre.vanrenssen@sydney.edu.au}{}{}
 \authorrunning{Ashvinkumar et al.} 

 \Copyright{Vikrant Ashvinkumar, Joachim Gudmundsson, Christos Levcopoulos, Bengt J. Nilsson and Andr\'e van Renssen}



 \nolinenumbers 

 \hideLIPIcs  

\begin{document}

\maketitle

\begin{abstract}
Online routing in a planar embedded graph is central to a number of fields and has been studied extensively in the literature. For most planar graphs no $O(1)$-competitive online routing algorithm exists. A notable exception is the Delaunay triangulation for which Bose and Morin~\cite{bose2004online} showed that there exists an online routing algorithm that is $O(1)$-competitive. However, a Delaunay triangulation can have $\Omega(n)$ vertex degree and a total weight that is a linear factor greater than the weight of a minimum spanning tree.

We show a simple construction, given a set $V$ of $n$ points in the Euclidean plane, of a planar geometric graph on $V$ that has small weight (within a constant factor of the weight of a minimum spanning tree on $V$), constant degree, and that admits a local routing strategy that is $O(1)$-competitive. Moreover, the technique used to bound the weight works generally for any planar geometric graph whilst preserving the admission of an $O(1)$-competitive routing strategy.

\end{abstract}

\section{Introduction}
\label{sec:Introduction}

The aim of this paper is to design a graph on $V$ (a finite set of points in the Euclidean plane) that is cheap to build and easy to route on. Consider the problem of finding a route in a geometric graph from a given source vertex $s$ to a given target vertex $t$. Routing in a geometric graph is a fundamental problem that has received
considerable attention in the literature.  In the offline setting, when we have full knowledge of the graph, the problem is well-studied and numerous algorithms exist for finding shortest paths (for example, the classic Dijkstra's Algorithm~\cite{Dijkstra}). In an online setting the problem becomes much more complex. The route is constructed incrementally and at each vertex a local decision has to be taken to decide which vertex to forward the message to. Without knowledge of the full graph, an online routing algorithm cannot identify a shortest path in general; the goal is to follow a path whose length approximates that of the shortest path.

Given a source vertex $s$, a target vertex $t$, and a message $m$, the aim is for an online routing algorithm to send $m$ together with a header $h$ from $s$ to $t$ in a graph $G$. Initially the algorithm only has knowledge of $s$, $t$ and the neighbors of $s$, denoted $\N(s)$. Note that it is commonly assumed that for a vertex $v$, the set $\N(v)$ also includes information about the coordinates of the vertices in $\N(v)$. Upon receiving a message $m$ and its header $h$, a vertex $v$ must select one of its neighbours to forward the message to as a function of $h$, $\N(v)$, $s$, and $t$. This procedure is repeated until the message reaches the target vertex $t$. Different routing algorithms are possible depending on the size of $h$ and the part of $G$ that is known to each vertex. Usually, there is a trade-off between the amount of information that is stored in the header and the amount of information that is stored in the vertices.

Bose and Morin~\cite{bose2004online} showed that greedy routing always reaches the intended destination on Delaunay triangulations. Dhandapani~\cite{D10} proved that every triangulation can be embedded in such a way that it allows greedy routing and Angelini et al.~\cite{AFG10} provided a constructive proof.

However, the above papers only prove that a greedy routing algorithm will succeed on the specific graphs therein. No attention is paid to the quality or \emph{competitiveness} of the resulting path relative to the shortest path. Bose and Morin~\cite{bose2004online} showed that many local routing strategies are not competitive but also show how to route competitively in a Delaunay triangulation. Bonichon~\etal~\cite{bonichon2018improved,bonichon2017upper} provided different local routing algorithms for the Delaunay triangulation, decreasing the competitive ratio, and Bonichon~\etal~\cite{bonichon2016gabriel} designed a competitive routing algorithm for Gabriel triangulations. 

To the best of our knowledge most of the existing routing algorithms consider well-known graph classes such as triangulations and $\Theta$-graphs. However, these graphs are generally very expensive to build. Typically, they have high degree $(\Omega(n))$ and the total length of their edges can be as bad as $\Omega(n)$ times that of the minimum spanning tree of $V$. 

On the other hand, there is a large amount of research on constructing geometric planar graphs with `good' properties. However, none of these have been shown to have all of bounded degree, weight, planarity, and the admission of competitive local routing. Bose~\etal~\cite{bfrv-olrdt-15} come tantalisingly close by providing a local routing algorithm for a plane bounded-degree spanner.

In terms of bounded degree, the best bound for plane spanners is 4 by Bonichon~\etal~\cite{Bonichon2015}. This spanner has a spanning ratio of 156.82. Another construction that also achieves a maximum degree of 4 was given by Kanj~\etal~\cite{DBLP:journals/corr/KanjPT16}, who reduced the spanning ratio to 20. In two special cases, Bose~\etal~\cite{DBLP:journals/corr/BiniazBCGMS16} showed that reducing the degree to 3 is possible. In terms of lower bounds, Dumitrescu and Ghosh~\cite{dumitrescu2016lower} showed that there exist point sets that require a spanning ratio of at least 1.4308. They also strengthened this bound to 2.1755 for spanners of degree 4 and 2.7321 for spanners of degree~3. 

The search for low weight spanners started in 1993 when Alth{\"{o}}fer et al.~\cite{addjs-93} presented the greedy spanner. Das et al.~\cite{dhn-93,dns-95} showed that the weight of a greedy spanner for a set $V$ of points in $R^d$ is within a constant factor times the weight of a minimum spanning tree, for any constant $d$. For a complete proof see the book by Narasimhan and Smid~\cite{narasimhan2007geometric}. In more recent work these results have been generalised to a wider family of doubling metrics~\cite{blw-19,fs-20,g-15}. 

In this paper we consider the problem of constructing a geometric graph of small weight and small degree that guarantees a local routing strategy that is $O(1)$-competitive. More specifically we show:

Given a set $V$ of $n$ points in the plane, together with two parameters $0 < \theta < \pi/2$ and $r > 0$, we show how to construct in $O(n \log n)$ time a planar $((1+1/r) \cdot \tau)$-spanner with degree at most $5\lceil 2\pi/\theta \rceil$, and weight at most $((2r+1)\cdot \tau)$ times the weight of a minimum spanning tree of $V$, where $\tau=1.998 \cdot \max(\pi/2,\pi \sin(\theta/2) + 1)$. This construction admits an $O(1)$-memory deterministic $1$-local routing algorithm with a routing ratio of no more than $5.90 \cdot (1 + 1/r) \cdot \max(\pi/2,\pi \sin(\theta/2) + 1)$. 

While we focus on our construction, we note that the techniques used to bound the weight of the graph apply generally to any planar geometric graph. In particular, using techniques similar to the ones we use, it may be possible to extend the results by Bose~\etal~\cite{bfrv-olrdt-15} to obtain other routing algorithms for bounded-degree light spanners.

\section{Building the Network} \label{sec:pruning}
Given a Delaunay triangulation $\mathcal{DT}(V)$ of a point set $V$ we will show that one can remove edges from $\mathcal{DT}(V)$ such that the resulting graph $\mathcal{BDG}(V)$ has constant degree and constant stretch-factor. We will also show that the resulting graph has the useful property that for every Delaunay edge $(u,v)$ in $\mathcal{DT}(V)$ there exists a spanning path along the boundary of the face in $\mathcal{BDG}(V)$ containing $u$ and $v$. This property will be critical to develop the routing algorithm in Section~\ref{sec:routing}. In Section~\ref{sec:lightness} we will show how to prune $\mathcal{BDG}(V)$ further to guarantee the lightness property while still being able to route in it.

\subsection{Building a Bounded Degree Spanner} \label{ssec:algorithm} 
The idea behind the construction is slightly reminiscent to that of the $\Theta$-graph: For a given parameter $0 < \theta < \pi/2$, let $\kappa = \Ceil{2\pi/\theta}$ and let $\mathcal{C}_{u,\kappa}$ be a set of $\kappa$ disjoint cones partitioning the plane, with each cone having angle measure at most $\theta$ at apex $u$. Let $v_0,\dots,v_m$ be the clockwise-ordered Delaunay neighbours of $u$ within some cone $C\in\mathcal{C}_{u,\kappa}$ (see Figure~\ref{fig:deg1}a).

\begin{figure}[ht] 
    \centering
    \includegraphics[width=\textwidth]{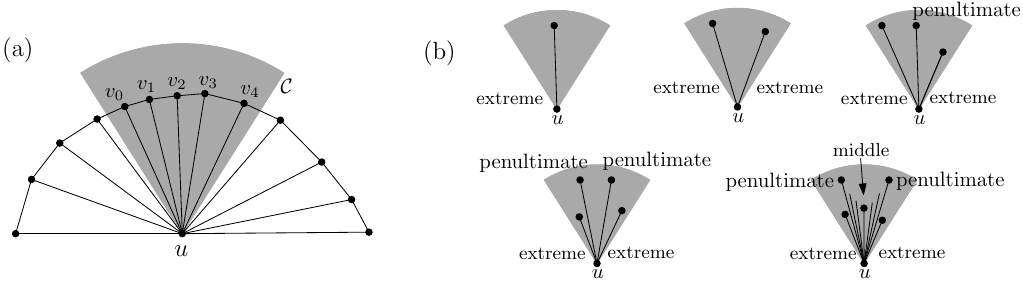}
    \caption{(a) An example of the vertices in some cone $C$ with apex $u$. (b) Extreme, penultimate, and middle are mutually exclusive properties taking precedence in that order.}
    \label{fig:deg1}
\end{figure}

If there is at least one edge at $u$ induced by $C$, call edges $uv_0$ and $uv_m$ \emph{extreme} at $u$.
Call edges $uv_1$ and $uv_{m-1}$ \emph{penultimate} at $u$ if there are two distinct extreme edges at $u$ induced by $C$ and at least one other edge at $u$ induced by $C$.
If there are two distinct edges that are extreme at $u$ induced by $C$, and two distinct edges that are penultimate at $u$ induced by $C$, and at least one other edge at $u$ induced by $C$, then, of the remaining edges incident to $u$ and contained in $C$, the shortest one is called a \emph{middle} edge at $u$ (see Figure~\ref{fig:deg1}b).

The construction removes every edge except the extreme, penultimate, and middle ones in every $C\in\mathcal{C}_{u,\kappa}$, for every point $u$, in any order.
The edges present in the final construction are thus the ones which are either extreme, penultimate, or middle at both of their endpoints (not necessarily the same at each endpoint). 

The resulting graph is denoted by $\mathcal{BDG}(V)$. The construction time of this graph is dominated by constructing the Delaunay triangulation, which requires $O(n \log n)$ time. Given the Delaunay triangulation, determining which edges to remove takes linear time (see Section~\ref{sec:construction1}). The degree of $\mathcal{BDG}(V)$ is bounded by $5\kappa$, since each of the $\kappa$ cones $C \in \mathcal{C}_{u,\kappa}$ can induce at most five edges. It remains to bound the spanning ratio.

\subsection{Spanning Ratio} \label{ssec:Algorithm-Spanning Ratio}
Before proving that the network is a spanner (Corollary~\ref{cor:desideratum1}) we will need to prove some basic properties regarding the edges in $\mathcal{BDG}(V)$. We start with a simple but crucial observation about consecutive Delaunay neighbours of a vertex $u$.

\begin{lemma}\label{lem:fat}
    Let $C$ be a cone with apex $u$ and angle measure $0 < \theta < \pi/2$. Let $v_l , v , v_r$ be consecutive clockwise-ordered Delaunay neighbours of $u$ contained in $C$. The interior angle $\angle(v_l , v , v_r)$ must be at least $\pi - \theta$.
\end{lemma}
\begin{proof}
 In the case when $\angle(v_l, v, v_r)$ is reflex in the quadrilateral $u,v_l,v,v_r$ the lemma trivially holds. Let us thus examine the case when $\angle(v_l, v, v_r)$ is not, in which case the quadrilateral $u,v_l,v,v_r$ is convex and, as $u,v_l,v$ and $u,v_r,v$ are Delaunay triangles, $\angle(v_l, u, v_r) + \angle(v_l, v, v_r)$ must be at least $\pi$ (see Figure~\ref{fig:fig2}a). Since $v_l$ and $v_r$ lie in a cone with apex $u$ of angle measure $\theta$, $\angle(v_l, u, v_r)$ is at most $\theta$. Hence, $\angle(v_l, v, v_r)$ is at least $\pi-\theta$. 
\end{proof}

\begin{figure}[ht] 
    \centering
    \includegraphics[width=10cm]{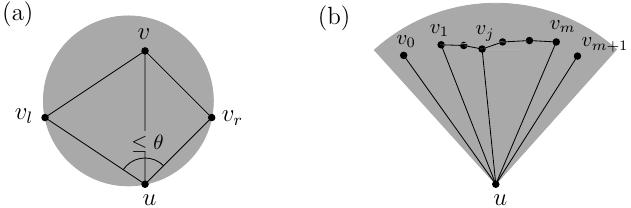}
    \caption{(a) Example placement of $u, v_l, v_r$ and $v$ in the circle $\circ(v_l,u,v_r)$ (b) The path from $v_1$ to $v_m$ along the Delaunay neighbours of $u$ must be in $\mathcal{BDG}(V)$. Furthermore, $uv_0$ and $uv_{m+1}$ are extreme, $uv_1$ and $uv_m$ are penultimate, and $uv_j$ is a middle edge.}
    \label{fig:fig2}
\end{figure}

This essentially means that $\angle(v_l,v,v_r)$ is wide, and will help us to argue when $v_lv$ and $vv_r$ must be in $\mathcal{BDG}(V)$ (Lemma~\ref{lem:boundary}). Next, we define \emph{protected}, \emph{fully protected}, and \emph{semi-protected} edges.   

\begin{definition}
    An edge $uv$ is \emph{protected} at $u$ (with respect to some fixed $\mathcal{C}_{u,\kappa}$) if it is extreme, penultimate, or middle at $u$. An edge $uv$ is \emph{fully protected} if it is protected at both $u$ and $v$. An edge $uv$ is \emph{semi-protected} at $u$ if it is protected at $u$ but not protected at $v$.
\end{definition}

Hence, an edge is contained in $\mathcal{BDG}(V)$ if and only if it is fully protected.
We continue with an observation that allows us to argue which edges are fully protected. 
\begin{observation} \label{obs:abut}
    If an edge $uv_i$ is not extreme at $u$, then $u$ must have consecutive clockwise-ordered Delaunay neighbours $v_{i-1}, v_i, v_{i+1}$, all in the same cone $C \in \mathcal{C}_{u,\kappa}$.
    Similarly, if $uv_i$ is neither extreme nor penultimate at $u$, then $u$ must have consecutive clockwise-ordered Delaunay neighbours $v_{i-2} , v_{i-1} , v_i , v_{i+1} , v_{i+2}$, all in the same cone $C \in \mathcal{C}_{u,\kappa}$.
\end{observation}

\begin{lemma}\label{lem:pensmall}
    Every edge that is penultimate or middle at one of its endpoints is fully protected.
\end{lemma}
\begin{proof}
    Consider an edge $uv$ that is penultimate or middle at $u$. Since it is protected at $u$, we need to show that it is protected at $v$. Since $uv$ is not extreme at $u$, $u$ must have consecutive clockwise-ordered Delaunay neighbours $v_l , v , v_r$ in the same cone by Observation~\ref{obs:abut}.

    We show that $uv$ must be extreme at $v$. Suppose for a contradiction that $uv$ is not extreme at $v$. Then, by Observation~\ref{obs:abut}, $v_lv$ and $vv_r$ are contained in the same cone with apex $v$ and angle at most $\theta < \pi/2$. However, by Lemma~\ref{lem:fat}, $\angle(v_l , v , v_r) \geq \pi - \theta > \theta$, which is impossible. Thus, $uv$ is extreme at $v$ and protected at $v$. Hence, the edge is fully protected. 
\end{proof}

Now we can argue about the Delaunay neighbours of a vertex (see Figure~\ref{fig:fig2}b for an illustration of the lemma).

\begin{lemma}\label{lem:boundary}
    Let $v_0 , \dots , v_{m+1}$ be the clockwise-ordered Delaunay neighbours of $u$ contained in some cone $C \in \mathcal{C}_{u,\kappa}$.
    The edges in the path $v_1 , \dots , v_m$ are all fully protected.
\end{lemma}
\begin{proof}
    Let $v_iv_{i+1}$ be an edge along this path for some $1 \leq i < m$.
    Suppose for a contradiction that $v_iv_{i+1}$ is not protected at $v_i$.
    It is thus, in particular, neither extreme nor  penultimate at $v_i$.
    Then, by Observation~\ref{obs:abut}, $v_iu$ and $v_iv_{i-1}$ must be contained in the same cone with apex $v_i$ as $v_iv_{i+1}$.
    By Lemma~\ref{lem:fat}, $\angle(v_{i-1} , v_i , v_{i+1}) \geq \pi - \theta > \theta$, contradicting that $v_iv_{i-1}$ and $v_iv_{i+1}$ lie in the same cone with apex $v_i$.
    The edge $v_iv_{i+1}$ must therefore be either extreme or penultimate, and thus protected, at $v_{i}$ for $i \geq 1$.
    An analogous argument shows that $v_iv_{i+1}$ is either extreme or penultimate at $v_{i+1}$ for $1 < i+1 \leq m$.
    It is thus fully protected. 
\end{proof}

Since these paths $v_1,\dots,v_m$ are included in $\mathcal{BDG}(V)$, we can modify the proof of Theorem~$3$ by Li and Wang~\cite{li2003efficient} to suit our construction to prove that $\mathcal{BDG}(V)$ is a spanner. 
 
\begin{theorem}\label{thm:bdg_del}
    $\mathcal{BDG}(V)$ is a $\max(\pi/2 , \pi \sin(\theta / 2) + 1)$-spanner of the Delaunay triangulation $\mathcal{DT}(V)$ for an adjustable parameter $0 < \theta < \pi/2$.
\end{theorem}
\begin{proof}
    The proof of this theorem is illustrated in Figure~\ref{fig:fig3}. We show that for any edge $uv$ in $\mathcal{DT}(V)$ that is not present in $\mathcal{BDG}(V)$, there is a spanning path in $\mathcal{BDG}(V)$ from $u$ to $v$.

    The edges in $\mathcal{BDG}(V)$ are exactly the edges in $\mathcal{DT}(V)$ that are fully protected.
    Without loss of generality, let $uv$ be an edge in $\mathcal{DT}(V)$ that is not protected at $u$.
    Then, $uv$ is not extreme and must be a chord of the face $u , v_0 , \dots , v_i = v , \dots , v_m$ where $uv_0$ is a middle edge and $uv_m$ is a penultimate edge. According to Lemma~\ref{lem:pensmall} $uv_0$ and $uv_m$ are edges in $\mathcal{BDG}(V)$, and according to Lemma~\ref{lem:boundary} all the edges in the path $v_0, \ldots , v_m$ are included in $\mathcal{BDG}(V)$.
    Moreover, $\angle(v_0 , u , v_i) < \angle(v_0 , u , v_m) < \theta < \pi/2$.
    Consider $S(v_0 , v_i)$, the shortest curve with endpoints $v_0$ and $v_i$ contained in the polygon $u, v_0 , \dots , v_i = v$.
    Label $|uv_0|$ with $x$, $|uv_i|$ with $y$, and let $w$ be the point on the segment $uv_i$ with length $x$ so that $|wv_i| = y - x$.

    We will show that $S(v_0 , v_i)$ is contained in the triangle $v_0 , w , v_i$.
    If none of $v_1, \dots, v_{i-1}$ are contained in the triangle $v_0 , w , v_i$, the claim must hold since all such vertices must be additionally outside the circle with centre $u$ and radius $x$ ($uv_0$ is the middle edge) and thus the line segment joining $v_0$ and $v_i$ is unobstructed.
    If any of $v_1, \dots, v_{i-1}$ are in the triangle $v_0 , w , v_i$, then $v_0$ must connect directly to one of them along $S(v_0,v_i)$, say $p$, and $v_i$ must connect directly to one of them, say $q$ possibly the same as $p$.
    Since $S(v_0,v_i)$ can be seen as the lower convex hull of $v_0, \dots, v_i$,
    and since $p$ and $q$ are in the triangle $v_0 , w , v_i$, the subpath of $S(v_0 , v_i)$ with endpoints $p$ and $q$ must be in the triangle $v_0 , w , v_i$ too.

    Since $S(v_0 , v_i)$ is convex with base $v_0v_i$ and contained in the triangle $v_0 , w , v_i$; it must thus have a length not more than $|v_0w| + |wv_i| = 2x\sin(\varphi/2) + y - x$ where $\varphi < \theta < \pi/2$ is the angle $\angle(v_0 , u , v_i)$.
    Now consider an edge of $S(v_0,v_i)$, say $v_k v_l$.
    The edge $v_k v_l$ shortcuts the subpath $v_k,\dots,v_l$ of $v_0,\dots,v_i$ in $\mathcal{BDG}(V)$.
    
   Dobkin et al.~\cite{DFS-90} (see also Lemma 3.3 in~\cite{bose2004online}) showed that the length $|v_k,\dots,v_l|$ is at most $\pi/2 \cdot |v_kv_l|$, provided that 
   \begin{enumerate}
     \item the straight-line segment between $v_k$ and $v_l$ lies outside the Voronoi region induced by $u$, and 
     \item the path $v_k,\dots,v_l$ lies on one side of the line through $v_k$ and $v_l$.
   \end{enumerate} 
    
    The first property follows from the fact that $\theta<\pi/2$ and the second property follows from the construction. Since both conditions hold, $|v_0,\dots,v_i| \leq |S(v_0,v_i)| \cdot \pi/2 \leq (|v_0w| + |wv_i|)\pi/2 = (2x\sin(\varphi/2) + y - x)\pi/2$.
    Putting everything together, we have that the path $u , v_0 , \dots , v_m$ has length at most
    \begin{align*}
        & x + (2x\sin(\varphi/2) + y - x)\pi/2\\
        =\; & y(\pi/2 + (\pi\sin(\varphi/2) + 1 - \pi/2)x/y)\\
        \leq\; & y(\pi/2 + (\pi\sin(\theta/2) + 1 - \pi/2)x/y)\\
        \leq\; & y\cdot\max(\pi/2 , \pi\sin(\theta/2) + 1)\\
        =\; & |uv| \cdot \max(\pi/2 , \pi\sin(\theta/2) + 1).
    \end{align*}
    Since $x/y \in (0,1)$, the last inequality immediately follows.

    The right-hand-most side of the inequality shows that for any edge $uv$ in $\mathcal{DT}(V)$, there is a $\max(\pi/2 , \pi\sin(\theta/2) + 1)$-spanning path in $\mathcal{BDG}(V)$ between $u$ and $v$.
    $\mathcal{BDG}(V)$ is thus a $\max(\pi/2 , \pi\sin(\theta/2) + 1)$-spanner of the Delaunay triangulation $\mathcal{DT}(V)$ for an adjustable parameter $0 < \theta < \pi/2$. 
\end{proof}

\begin{figure} 
    \centering
    \includegraphics[width=10cm]{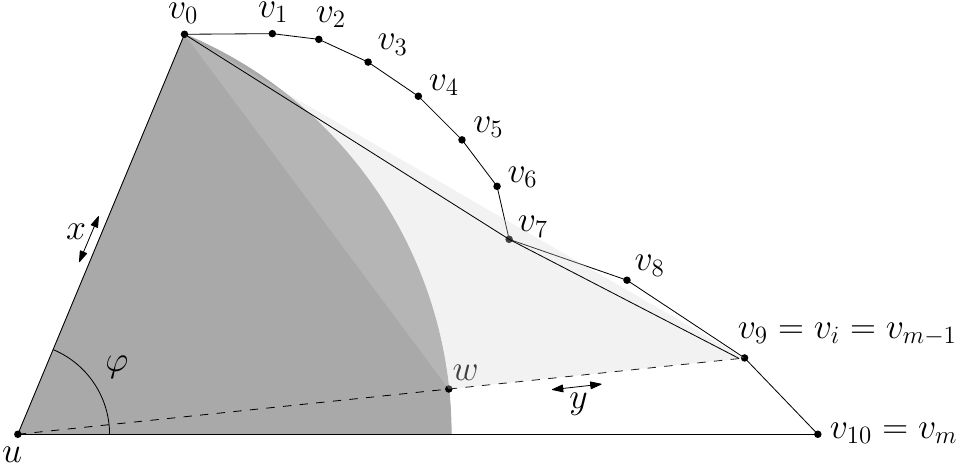}
    \caption{Illustrating Theorem~\ref{thm:bdg_del}. The spanning path from $u$ to $v$ is $u, v_0, v1, \dots, v_9 = v$. Note that the path $v_0, v_7, v$ is indeed contained in triangle $v_0, w, v_i$.}
    \label{fig:fig3}
\end{figure}

Note that from the proof of Theorem~\ref{thm:bdg_del} it follows that for every Delaunay edge $uv$  that is \emph{not} in $\mathcal{BDG}(V)$, there is a path from $u$ to $v$ along the face of $\mathcal{BDG}(V)$ containing $uv$ realising a path of length at most $\max(\pi/2 , \pi \sin(\theta / 2) + 1) \cdot |uv|$. This is a key observation that will be used in Section~\ref{sec:routing}.

\subsection{Algorithmic Construction of $\mathcal{BDG}(V)$}
\label{sec:construction1}
For completeness we state the algorithm in Section~\ref{sec:pruning} as pseudocode and analyse its time complexity.
\begin{algorithm}
    \caption{$\mathcal{BDG}(V)$}
    \begin{algorithmic}[1]
        \Require $V$
        \Require $0 < \theta < \pi/2$
        \State $E \gets \{\}$
        \State $\mathcal{DT} \gets \mathcal{DT}(V)$
        \For{$u \in V$}
        \State Compute $\mathcal{C}_{u,\kappa}$, where $\kappa = \Ceil{2\pi/\theta}$.
        \EndFor
        \For{$u \in V$}
            \For{$uv \in E(\mathcal{DT})$}
                \State Bucket $uv$ into $C \in \mathcal{C}_{u,\kappa}$.
            \EndFor
        \EndFor
        \For{$u \in V$}
            \For{$C \in \mathcal{C}_{u,\kappa}$}
                \State Reset values of $e_1, e_2, p_1, p_2, m$.
                \For{$e$ bucketed into $C$}
                    \State $e_1 \gets \text{arg}\min_{angle}(e_1, e)$
                    \State $e_2 \gets \text{arg}\max_{angle}(e_2, e)$
                \EndFor
                \For{$e$ bucketed into $C \backslash \{e_1,e_2\}$}
                    \State $p_1 \gets \text{arg}\min_{angle}(p_1, e)$
                    \State $p_2 \gets \text{arg}\max_{angle}(p_2, e)$
                \EndFor
                \For{$e$ bucketed into $C \backslash \{e_1,e_2,p_1,p_2\}$}
                    \State $m \gets \text{arg}\min_{length}(m, e)$
                \EndFor
                \State Mark $e_1,e_2, p_1,p_2, m$, if their values are set, as protected by $u$.
            \EndFor
        \EndFor
        \For{$uv \in E(\mathcal{DT})$}
            \If{$uv$ marked as protected by both endpoints}
                \State $E = E \cup \{uv\}$.
            \EndIf
        \EndFor
        \Return $(V,E)$.
    \end{algorithmic}
\end{algorithm}

\begin{theorem}
    $\mathcal{BDG}(V)$ takes $O(n\log n)$ time to construct. $\mathcal{BDG}(V)$ takes $O(n)$ time to construct if the input is a Delaunay triangulation $\mathcal{DT}(V)$ on $V$.
\end{theorem}
\begin{proof}
    The construction of the Delaunay triangulation $\mathcal{DT}(V)$ at line \texttt{2} takes $O(n \log n)$ time.

    The loops at lines \texttt{3,5,8,20} are independent of each other.
    The one starting at line \texttt{3} takes $O(n)$ time and the one on line
    \texttt{5} takes $O(n)$ time since there are a linear number of edges in $E(\mathcal{DT})$, which we look at twice (once for each endpoint), and the bucketing of each edge takes $\kappa$ time at most.
    The loop starting at line \texttt{8} takes $O(n)$ time since there are a linear number of edges in $E(\mathcal{DT})$, which we look at six times at most (thrice for each endpoint).
    Finally, the loop at line \texttt{20} takes $O(n)$ time since there are a linear number of edges in $E(\mathcal{DT})$.

    The result follows that
    $\mathcal{BDG}(V)$ takes $O(n\log n)$ time to construct and
    $\mathcal{BDG}(V)$ takes $O(n)$ time to construct if the input is a Delaunay triangulation $\mathcal{DT}(V)$ on $V$. 
\end{proof}

Putting the results from this section together, using that the Delaunay triangulation is a $1.998$-spanner~\cite{xia2013stretch}, and observing that $\mathcal{BDG}(V)$ is trivially planar since it is a subgraph of the Delaunay triangulation, we obtain:

\begin{corollary}\label{cor:desideratum1}
    Given a set $V$ of $n$ points in the plane and a parameter $0 < \theta < \pi/2$, one can in $O(n \log n)$ time compute a graph $\mathcal{BDG}(V)$ that is a planar $\tau$-spanner having degree at most $5\lceil 2\pi/\theta \rceil$, where $\tau=1.998 \cdot \max(\pi/2,\pi \sin(\theta/2) + 1)$.
\end{corollary}

\section{Routing} \label{sec:routing}
In order to route efficiently on $\mathcal{BDG}(V)$, we modify Bonichon~\etal's routing algorithm~\cite{bonichon2017upper} on the Delaunay Triangulation. Given a source $s$ and a destination $t$ on the Delaunay triangulation $\mathcal{DT}(V)$, we assume without loss of generality that the line segment $[st]$ is horizontal with $s$ to the left of $t$. 
Bonichon~\etal's routing algorithm~\cite{bonichon2017upper} then works as follows: When we are at a vertex $v_i$ ($v_0 = s$), set $v_{i+1}$ to $t$ and terminate if $v_it$ is an edge in $\mathcal{DT}(V)$. Otherwise, consider the rightmost Delaunay triangle $T_i = v_i, p, q$ at $v_i$ that has a non-empty intersection with $[st]$. Denote the circumcircle $\circ(v_i,p,q)$ with $C_i$, denote the leftmost point of $C_i$ with $w_i$, and the rightmost intersection of $C_i$ and $[st]$ with $r_i$. 
\begin{itemize}
    \item If $v_i$ is encountered in the clockwise walk along $C_i$ from $w_i$ to $r_i$, set $v_{i+1}$ to $p$, the first vertex among $\{p,q\}$ encountered on this walk starting from $v_i$ (see Figure~\ref{fig:route:decision}a). 
    \item Otherwise, set $v_{i+1}$ to $q$, the first vertex among $\{p,q\}$ to be encountered in the counterclockwise walk along $C_i$ starting from $v_i$ (see Figure~\ref{fig:route:decision}b). 
\end{itemize}

\begin{figure}[ht] 
    \centering
    \includegraphics{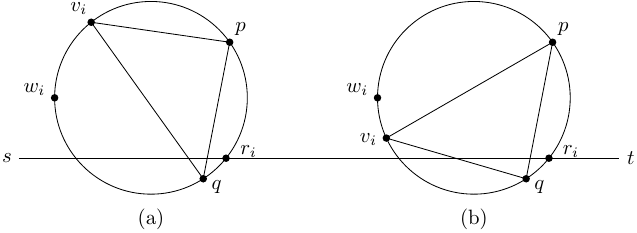}
    \caption{The routing choice: (a) At $v_i$ we follow the edge to $p$. (b) At $v_i$ we follow the edge to $q$.}
    \label{fig:route:decision}
\end{figure}

We relax Bonichon~\etal's routing algorithm~\cite{bonichon2017upper} in such a way that it no longer necessarily uses the rightmost intersected triangle: At $v_0$, we set $A_0 = T_0$; at $v_{i}$ for $i > 0$, we will find a Delaunay triangle $A_i$ based on the Delaunay triangle $A_{i-1} = v_{i-1}, x, y$ used in the routing decision at $v_{i-1}$, where one of $x$ or $y$ is $v_i$.

Let $A_i = v_i, p, q$ be any Delaunay triangle with a non-empty intersection with $[st]$ to the right of the intersection of $A_{i-1}$ with $[st]$ and which, moreover, satisfies the condition that if $v_i$ is above $[st]$, then, when making a counterclockwise sweep centred at $v_i$ starting from $v_iv_{i-1}$, we encounter $v_iq$ before $v_ip$, with $v_iq$ intersecting $[st]$ and $v_ip$ not intersecting $[st]$. Figure~\ref{fig:down_up} illustrates two concrete examples of $A_i$ given $A_{i-1}$. An analogous statement for choosing $A_i$ holds when $v_i$ lies below $[st]$, sweeping in clockwise direction. 

\begin{figure}[ht] 
    \centering
    \includegraphics{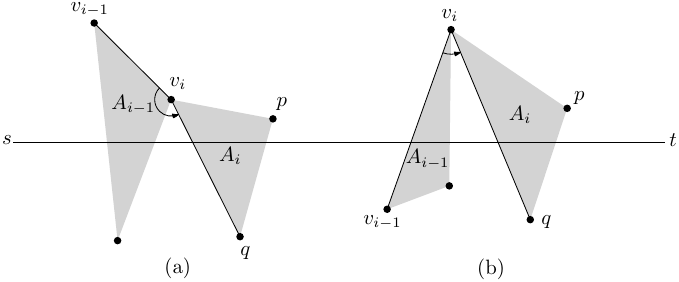}
    \caption{Candidate triangles $A_i$ given $A_{i-1}$: (a) when $[st]$ is not crossed when moving from $v_{i-1}$ to $v_i$, (b) when $[st]$ is crossed when moving from $v_{i-1}$ to $v_i$.}
    \label{fig:down_up}
\end{figure}

We note that these triangles $A_i$ always exist, since the rightmost Delaunay triangle intersecting $[st]$ is a candidate. Furthermore, the triangles occur in order along $[st]$ by definition. This implies that the relaxation of Bonichon~\etal's routing algorithm~\cite{bonichon2017upper} terminates. 

\begin{theorem}\label{thm:ccra_route}
    The relaxation of Bonichon~\etal's routing algorithm~\cite{bonichon2017upper} on the Delaunay triangulation is 1-local and  has a routing ratio of at most $(1.185043874 + 3\pi/2) \approx 5.90$.
\end{theorem}
\begin{proof}
    The 1-locality follows by construction. The proof for the routing ratio of Bonichon~\etal's routing algorithm~\cite{bonichon2017upper} holds for 
    its relaxed version, since the only parts of their proof using the property that $T_i$ is rightmost are: 
    \begin{enumerate}
        \item The termination of the algorithm (which we argued above).
        \item The categorisation of the Worst Case Circles of Delaunay triangles $T_i$ into three mutually exclusive cases (which we discuss next).
    \end{enumerate}
    Thus, the relaxation of Bonichon~\etal's routing algorithm~\cite{bonichon2017upper} on the Delaunay triangulation has a routing ratio of at most $(1.185043874 + 3\pi/2) \approx 5.90$. 
\end{proof}

\subsection{Worst Case Circles}
In the analysis of the routing ratio of Bonichon~\etal's routing algorithm~\cite{bonichon2017upper}, the notion of Worst Case Circles is introduced whereby the length of the path yielded by the algorithm is bounded above by some path consisting of arcs along these Worst Case Circles; this arc-path is then shown to have a routing ratio of $5.90$.

Suppose we have a candidate path, and are given a Delaunay triangle $v_i, v_{i+1}, u$ intersecting $[st]$; we denote its circumcircle by $C_i$ with centre $O_i$. 
The Worst Case Circle $C'_i$ is a circle that goes through $v_i$ and $v_{i+1}$, 
whose centre $O'_i$ is obtained by starting at $O_i$ and moving it along the perpendicular 
bisector of $[v_iv_{i+1}]$ until either $st$ is tangent to $C'_i$ or $v_i$ is the leftmost point 
of $C'_i$, whichever occurs first. 
The direction $O'_i$ is moved towards depends on the routing decision at $v_i$: 
if $v_i$ is encountered on the clockwise walk from $w_i$ to $r_i$, then $O'_i$ is moved towards this arc, 
and otherwise, $O'_i$ is moved towards the opposite direction. 
Letting $w'_i$ be the leftmost point of $C'_i$, we can categorise the Worst Case Circles into the following three mutually exclusive types (see Figure~\ref{fig:circles}): 
\begin{enumerate}
    \item Type $X_1$ : $v_i \neq w'_i$, and $[v_iv_{i+1}]$ does not cross $[st]$, and $st$ is tangent to $C'_i$.
    \item Type $X_2$ : $v_i = w'_i$ and $[v_iv_{i+1}]$ does not cross $[st]$.
    \item Type $Y$ : $v_i = w'_i$ and $[v_iv_{i+1}]$ crosses $[st]$.
\end{enumerate}

\begin{figure}[ht] 
    \centering
    \includegraphics{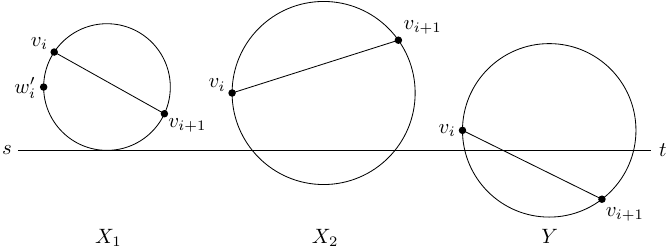}
    \caption{The three types of Worst Case Circles.}
    \label{fig:circles}
\end{figure}

Next, we show that the Worst Case Circles of Delaunay triangles $A_i$ fall into the same categories. 
Let $C_i$ be the circumcircle of $A_i$ centred at $O_i$, let $w_i$ be the leftmost point of $C_i$, and let $r_i$ be the right intersection of $C_i$ with $[st]$. 
We begin with the following observation which follows from how the criteria forces $A_i$ to intersect $[st]$: 

\begin{observation}\label{obs:routeorder}
    Let $A_i = v_i, p, q$. 
    Taking a clockwise walk along $C_i$ from $v_i$ to $r_i$, exactly one of $p$ or $q$ is encountered. An analogous statement holds for the counterclockwise walk. 
\end{observation}

This observation captures the necessary property that allows the categorisation to go through. 
We denote the Worst Case Circle of $A_i$ by $C'_i$ with centre $O'_i$, and leftmost point~$w'_i$. 

\begin{lemma}\label{lem:routecat}
$C'_i$ can be categorised into the following three mutually exclusive types: 
\begin{enumerate}
    \item Type $X_1$ : $v_i \neq w'_i$, and $[v_iv_{i+1}]$ does not cross $[st]$, and $st$ is tangent to $C'_i$.
    \item Type $X_2$ : $v_i = w'_i$ and $[v_iv_{i+1}]$ does not cross $[st]$.
    \item Type $Y$ : $v_i = w'_i$ and $[v_iv_{i+1}]$ crosses $[st]$.
\end{enumerate}
\end{lemma}
\begin{proof}
    If $[v_iv_{i+1}]$ does not cross $[st]$, $C'_i$ is clearly of type $X_1$ or $X_2$.

    Consider when $[v_iv_{i+1}]$ crosses $[st]$. 
    Without loss of generality, let $v_i$ be above $[st]$ and $v_{i+1}$ be below $[st]$. 
    By Observation~\ref{obs:routeorder}, $v_i$ occurs on the counterclockwise walk around $C_i$ from $w_i$ to $r_i$, 
    for if not, neither vertex of $A_i$ occurs on the clockwise walk around $C_i$ from $v_i$ to $r_i$. 
    Since $v_i$ is above $[st]$, it lies above the leftmost intersection of $C_i$ with $[st]$ and below~$w_i$. 
    
    Since $O'_i$ is moved along the perpendicular bisector of $[v_iv_{i+1}]$ towards the counterclockwise arc of $v_i$ to $v_{i+1}$, it must be that $w'_i$ (which starts at $w_i$ when $O'_i$ starts at $O_i$) moves onto $v_i$ eventually. 
    Thus, $C'_i$ is Type $Y$. 
\end{proof}

\subsection{Routing on $\mathcal{BDG}(V)$} \label{ssec:routing}
In order to route on $\mathcal{BDG}(V)$, we simulate the relaxation of Bonichon~\etal's routing algorithm~\cite{bonichon2017upper} described in the previous section. We first prove a property that allows us to distribute information about edges over their endpoints. 

\begin{lemma}\label{thm:evprotect}
    Every edge $uv \in \mathcal{DT}(V)$ is protected by at least one of its endpoints $u$ or $v$.
\end{lemma}
\begin{proof}
    Suppose that $uv$ is not protected at $u$.
    Then $uv$ is not extreme at $u$ and thus by Observation~\ref{obs:abut}, $u$ must have consecutive clockwise-ordered Delaunay neighbours $v_l, v, v_r$.
    By Lemma~\ref{lem:fat}, $\angle(v_l, v, v_r) \geq \pi - \theta > \theta$ since $0 < \theta < \pi/2$, and thus $v_l$ and $v_r$ cannot both belong to the same cone with apex $v$ and angle at most $\theta$.
    Since $v_r, u, v_l$ are consecutive clockwise-ordered Delaunay neighbours of $v$, and $vv_l$ and $vv_r$ cannot be in the same cone, it follows that $vu$ is extreme at $v$.
    Hence, $uv$ is protected at $v$ when it is not protected at $u$. 
\end{proof}

This lemma allows us to store all edges of the Delaunay triangulation by distributing them over their endpoints. At each vertex $u$, we store: 
\begin{enumerate}
    \item Fully protected edges $uv$, with two additional bits to denote whether it is extreme, penultimate, or middle at $u$.
    \item Semi-protected edges $uv$ (only protected at $u$), with one additional bit denoting whether the clockwise or counterclockwise face path is a spanning path to $v$.
\end{enumerate}

We can label the vertices of $\mathcal{BDG}(V)$ in this way, denoting this augmented graph as a Marked Bounded Degree Graph or $\mathcal{MBDG}(V)$ for short. Pseudocode and its running time analysis can be found in Appendix~\ref{app:MBDG}.

\begin{theorem}\label{thm:mbdg1}
    $\mathcal{MBDG}(V)$ stores $O(1)$ words of information at each of its vertices.
\end{theorem}
\begin{proof}
    According to Corollary~\ref{cor:desideratum1}, each vertex in $\mathcal{BDG}(V)$ is incident to at most $5\kappa$ (fully and semi) protected edges, where $\kappa$ is a fixed constant. From the above discussion a vertex may store 2 bits for each incident protected edge in $\mathcal{MBDG}(V)$, which immediately proves the theorem
\end{proof}

In the remaining part of Section~\ref{ssec:routing} we will focus our attention on routing in $\mathcal{MBDG}(V)$. When we write ``an edge is followed" or ``walking along a face" or any statement of that sort, this is always done in $\mathcal{MBDG}(V)$ using only the information stored in each vertex unless otherwise stated. At a high level, the routing algorithm on $\mathcal{MBDG}(V)$ works as follows: the simulation searches for a suitable candidate triangle $A_i$ at $v_i$, possibly taking a walk from $v_i$ along a face to be defined later in order to do so. 
Once $A_i$ has been found, we will know the locations of $v_i, p, q$, where $p$ and $q$ are candidate vertices for $v_{i+1}$, and we can thus use the routing criteria of Bonichon~\etal's routing algorithm~\cite{bonichon2017upper} to determine whether to route to $p$ or to route to $q$. Next, we describe how to route on the non-triangular faces of $\mathcal{MBDG}(V)$; the vertices of any such face can always be labelled with $v, u_1, \dots, u_m$ where $vu_1$ and $vu_m$ are a middle edge and a penultimate edge at $v$ (see Figure~\ref{fig:face}).

\begin{figure}[ht] 
    \centering
    \includegraphics{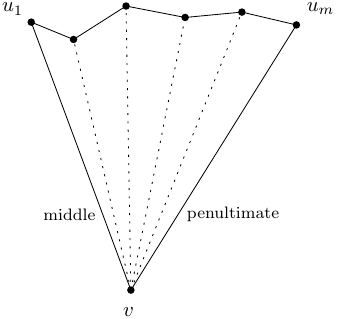}
    \caption{A non-triangular face of $\mathcal{MBDG}(V)$. Dotted edges here are unprotected at $v$.}
    \label{fig:face}
\end{figure}

There are two different situations wherein we must route on such faces. The first situation is when we want to move from $v$ to any other vertex (the destination vertex along such a face is undetermined until it is reached); we shall use Unguided Face Walks to ensure a face-route with constant stretch in such a situation. The second situation is when we want to move from any other vertex on this face to $v$; we shall use Guided Face Walks to ensure a face-route with constant stretch in such a situation

\subsubsection{Unguided Face Walks}
Suppose $vu_1$ and $vu_m$ are a middle edge and a penultimate edge and suppose that $vu_1$ is the shorter of the two. We want to route from $v$ to any other vertex $p$ on this face. For any such vertex $p$ on this face, we refer to the spanning face path from $v$ to $p$ starting with $vu_1$ as an Unguided Face Walk from $v$ to $p$.

\begin{figure}[ht] 
    \centering
    \includegraphics{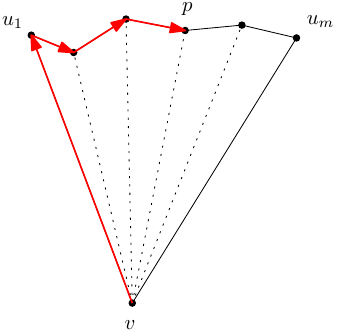}
    \caption{An Unguided Face Walk from $v$ to $p$. $vu_1$ and $vu_m$ are not labeled with ``middle" or ``penultimate" to emphasize that we can take the shorter of the two.}
    \label{fig:unguided}
\end{figure}

In the simulation, we use Unguided Face Walks in a way that $p$ is undetermined until it is reached; 
we will take an Unguided Face Walk from $v$ and test at each vertex along this walk if it satisfies some property, ending the walk if it does.  Routing in this manner from $v$ to $p$ can easily be done locally: 
Suppose $vu_1$ was counterclockwise to $vu_m$ (see Figure~\ref{fig:unguided}).  Then, at any intermediate vertex $u_{i}$, we take the edge immediately counterclockwise to $u_iu_{i-1}$ ($v = u_0$). 
The procedure when $vu_1$ is clockwise to $vu_m$ is analogous. 

\begin{observation}\label{obs:unguide2}
    An Unguided Face Walk needs $O(1)$ memory since at $u_i$, the previous vertex along the walk $u_{i-1}$ must be stored 
    in order to determine $u_{i+1}$.
\end{observation}

\begin{observation}\label{obs:unguide}
    An Unguided Face Walk from $v$ to $p$ has a stretch factor of at most $\max(\pi/2 , \pi \sin(\theta / 2) + 1)$ as shown in the proof of Theorem~\ref{thm:bdg_del}.
\end{observation}

\subsubsection{Guided Face Walks}
Suppose we want to route from $p$ to $v$ where $pv$ is extreme at $p$ but not protected at $v$ (i.e., it is a semi-protected edge stored at $p$). 
Then, $pv$ is a chord of some face determined by $vu_1$ and $vu_m$ where the former is a middle edge and the latter a penultimate edge. 
Moreover, recall that we stored a bit with the semi-protected edge $pv$ at $p$ indicating whether to take the edge clockwise or counterclockwise to reach $v$. 
We refer to the face path from $p$ to $v$ following the direction pointed to by these bits as the Guided Face Walk from $p$ to $v$ (see Figure~\ref{fig:guided}). 
Routing from $p$ to $v$ can now be done as follows:
\begin{enumerate}
    \item{At $p$, store $v$ in memory.}
    \item{Until $v$ is reached, if there is an edge to $v$, take it. 
        Otherwise, take the edge pointed to by the bit of the semi-protected edge to $v$.}
\end{enumerate}

\begin{figure}[ht] 
    \centering
    \includegraphics{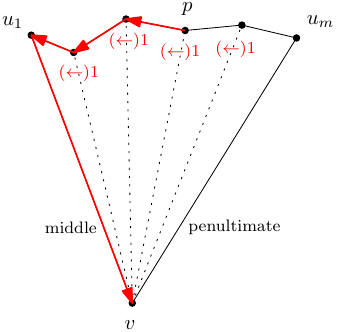}
    \caption{A Guided Face Walk from $p$ to $v$. Vertices at which an edge is semi-protected are labeled with the edge and a bit-direction.}
    \label{fig:guided}
\end{figure}

\begin{observation}\label{obs:guide2}
    A Guided Face Walk needs $O(1)$ memory since $v$ needs to be stored in memory for the duration of the walk.
\end{observation}

\begin{observation}\label{obs:guide}
    A Guided Face Walk from $p$ to $v$ has a stretch factor of at most $\max(\pi/2 , \pi \sin(\theta / 2) + 1)$ as shown in the proof of Theorem~\ref{thm:bdg_del}.
\end{observation}

\subsubsection{Simulating the Relaxation of Bonichon~\etal's Routing Algorithm~\cite{bonichon2017upper}}
We are now ready to describe the routing algorithm on $\mathcal{MBDG}(V)$ in more detail. First, we consider finding the first vertex after $s$. If $st$ is an edge, take it and terminate. 
Otherwise, at $s = v_0$, we consider all edges protected at $s$, and let $su_1$ and $su_m$ be the first such edge encountered in a counterclockwise and clockwise sweep starting from $[st]$ centred at $s$. 
There are two subcases.

(I) If both $su_1$ and $su_m$ are not middle edges at $s$, then $s, u_1, u_m$ is a Delaunay triangle $A_0$. Determine whether to route to $u_1$ or $u_m$, using the same criteria used in Bonichon~\etal's routing algorithm~\cite{bonichon2017upper} (see the beginning of Section~\ref{sec:routing}). If the picked edge is fully protected, we follow it. Otherwise, we take the Guided Face Walk from $s$ to this vertex. 

(II) If one of $su_1$ and $su_m$ is a middle edge at $s$, the other edge must then be a penultimate edge. 
Then, $A_0 = s,p,q$ must be contained in the cone with apex $s$ sweeping clockwise from $su_1$ to $su_m$. 
We assume that $su_1$ is shorter than $su_m$. 
Take the Unguided Face Walk from $s$ until some $u_i$ such that $u_i = p$ is above $[st]$ and $u_{i+1} = q$ is below $[st]$. 
We have now found $A_0 = s,p,q$ and we determine whether to route to $p$ or $q$, using the same criteria used in Bonichon~\etal's routing algorithm~\cite{bonichon2017upper} (see the beginning of Section~\ref{sec:routing}). 

In both cases, the memory used for the Face Walks is cleared and $A_0 = s, u_1, u_m$ or $A_0 = s, p, q$ is stored as the last triangle used.

Next, we focus on how to simulate a routing step from an arbitrary vertex $v_i$. 
Suppose $v_i$ is above $[st]$, and that $A_{i-1}$ is stored in memory. 
If $v_it$ is an edge, take it and terminate. 
Otherwise, let $v_if$ be the rightmost edge of $A_{i-1}$ that intersects $[st]$, and $\overline{v_if}$ be its extension to a line.
Make a counterclockwise sweep, centred at $v_i$ and starting at $v_if$, through all edges that are protected at $v_i$ that lie in the halfplane defined by $\overline{v_if}$ that contains $t$. 
Note that this region must have at least one such edge, since otherwise $v_if$ is a convex hull edge, which cannot be the case since $s$ and $t$ are on opposite sides. 

(I) If there is some edge that does not intersect $[st]$ in this sweep, let $v_iu_1$ be the first such edge encountered in the sweep and let $v_iu_m$ be the protected edge immediately clockwise to $v_iu_1$ at $v_i$. 
There are two cases to consider. 

(I.I) If $A_{i-1}$ is not contained in the cone with apex $v_i$ sweeping clockwise from $v_iu_1$ to $v_iu_m$ (see Figure~\ref{fig:simulation_1}a), simulating a step of the relaxation of Bonichon~\etal's Routing Algorithm~\cite{bonichon2017upper} is analogous to the method used for the first step: determine if $v_iu_1$ or $v_iu_m$ is a middle edge and use a Guided or Unguided Face Walk to reach the proper vertex of $A_i$. 

\begin{figure}[ht] 
    \centering
    \includegraphics{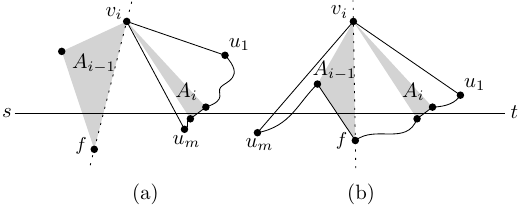}
    \caption{Simulating a step of the relaxation of Bonichon~\etal's routing algorithm~\cite{bonichon2017upper}: (a) Case~I.I, (b) case~I.II.}
    \label{fig:simulation_1}
\end{figure}

(I.II) If $A_{i-1}$ is contained in the cone with apex $v_i$ sweeping clockwise from $v_iu_1$ to $v_iu_m$ (see Figure~\ref{fig:simulation_1}b), then one of $v_iu_1$ and $v_iu_m$ must be a middle edge and the other a penultimate edge. This must be the case since the edge $v_if$ is contained in the interior of the cone with apex $v_i$ sweeping clockwise from $v_iu_1$ to $v_iu_m$ and is thus not protected at $v_i$; unprotected edges at $v_i$ are always between a middle and a penultimate edge. Then, $A_i = v_i,p,q$ must be contained in the cone with apex $v_i$ sweeping clockwise from $v_iu_1$ to $v_if$. 

We take the Unguided Face Walk, starting from the shorter of $v_iu_1$ and $v_iu_m$. If we start from $v_iu_1$, we stop when we have found some $u_i$ such that $u_i = p$ is above $[st]$ and $u_{i+1} = q$ is below $[st]$, and make the decision to complete the Unguided Face Walk to $q$ or not. If, on the other hand, we start from $v_iu_m$, we stop when we have both passed $f$ in the Unguided Face Walk (to ensure that $A_i$ lies to the right of $A_{i-1}$) and found some $u_{i+1}$ such that $u_{i+1} = q$ is below $[st]$ and $u_i = p$ is above $[st]$, and make the decision to complete the Unguided Face Walk to $p$ or not.

(II) If all of the edges in the sweep intersect $[st]$ (see Figure~\ref{fig:simulation_2}), let $v_iu_m$ be the last edge encountered in the sweep, and $v_iu_1$ be the protected edge immediately counterclockwise to it, which must be in the halfplane defined by $\overline{v_if}$ that does not contain $t$. 
Note that $A_{i-1}$ cannot be contained in this cone, as that would imply that $\angle(u_1,v_i,u_m) \geq \pi$, making $v_iu_m$ a convex hull edge. 
Simulating the Delaunay routing algorithm is analogous to the method used for the first step: determine if $v_iu_1$ or $v_iu_m$ is a middle edge and use a Guided or Unguided Face Walk to reach the proper vertex of $A_i$. 

\begin{figure}[ht] 
    \centering
    \includegraphics{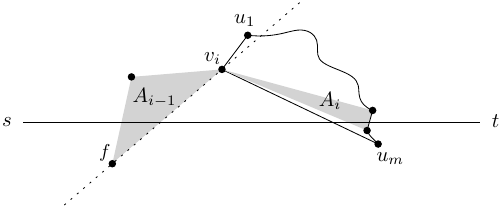}
    \caption{Simulating a step of the relaxation of Bonichon~\etal's routing algorithm~\cite{bonichon2017upper} (case~II).}
    \label{fig:simulation_2}
\end{figure}

In all cases, we clear the memory and store $A_i = v_i,p,q$ as the previous triangle. The case where $v_i$ lies below $[st]$ is analogous. We obtain the following theorem. 

\begin{theorem}\label{thm:route10}
    The simulation of the relaxation of Bonichon~\etal's routing algorithm~\cite{bonichon2017upper} on $\mathcal{MBDG}(V)$ is 1-local, has a routing ratio of at most $5.90 \cdot \max(\pi/2 , \pi \sin(\theta / 2) + 1)$ and uses $O(1)$ memory.
\end{theorem}

\section{Lightness} \label{sec:lightness}
In the previous sections we have presented a bounded degree network $\mathcal{MBDG}(V)$ with small spanning ratio that allows for local routing. It remains to show how we can prune this graph even further to guarantee that the resulting network $\mathcal{LMBDG}(V)$ also has low weight. 

We will describe a pruning algorithm that takes $\mathcal{MBDG}(V)$  and returns a graph (Light Marked Bounded Degree Graph) $\mathcal{LMBDG}(V) \subseteq \mathcal{MBDG}(V)$, allowing a trade-off between the weight (within a constant times that of the minimum spanning tree of $V$) and the (still constant) stretch factor. Then, we show how to route on $\mathcal{LMBDG}(V)$ with a constant routing ratio and constant memory. 

\subsection{The Levcopoulos and Lingas Protocol}
To bound the weight of $\mathcal{MBDG}(V)$, we use the algorithm by Levcopoulos and Lingas~\cite{levcopoulos1992there} with two slight modifications: (1) allow any planar graph as input instead of only Delaunay triangulations, and (2) marking the endpoints of pruned edges to facilitate routing. 

At a high level, the algorithm works as follows: Given $\mathcal{MBDG}(V)$, we compute its minimum spanning tree and add these edges to $\mathcal{LMBDG}(V)$.  We then take an Euler Tour around the minimum spanning tree, treating it as a degenerate polygon $P$ enclosing $V$.  Finally, we start expanding $P$ towards the convex hull $CH(V)$.  As edges of $\mathcal{MBDG}(V)$ enter the interior of $P$, we determine whether to add them to $\mathcal{LMBDG}(V)$. This decision depends on a given  parameter $r > 0$.  If an edge is excluded from $\mathcal{LMBDG}(V)$, we augment its endpoints with information to facilitate routing should that edge be used in the path found on $\mathcal{MBDG}(V)$.  Once $P$ has expanded into $CH(V)$, we return $\mathcal{LMBDG}(V)$. 

To explicate further upon this, let us first acknowledge and differentiate between a few kinds of edges which will play a part in the following discussion:
\begin{enumerate}
    \item Convex hull edges of $CH(V)$. 
    \item Boundary edges of the polygon $P$ that encloses $V$. 
    \item Included settled edges, which are edges of $\mathcal{MBDG}(V)$ in $P$ and included in $\mathcal{LMBDG}(V)$. 
    \item Excluded settled edges, which are edges of $\mathcal{MBDG}(V)$ in $P$ and excluded from $\mathcal{LMBDG}(V)$. 
    \item Unsettled edges, which are edges of $\mathcal{MBDG}(V)$ outside of $P$ and whose inclusion in $\mathcal{LMBDG}(V)$ have not yet been determined. 
\end{enumerate}
Note that while the last three kinds are mutually exclusive, there may be edges which are of more than one kind. 
For example, a boundary edge of $P$ can coincide with a convex hull edge of $CH(V)$. 
\subsubsection{How the Polygon Grows}
For each iteration of the Levcopoulos and Lingas Protocol, $P$, a polygon without holes, grows, consuming more area and more edges of $\mathcal{MBDG}(V)$, until it coincides completely with the convex hull $CH(V)$.  Let us consider a single iteration of the algorithm. 

Consider any edge $uv$ on the convex hull $CH(V)$. If part of the boundary of $P$ coincides with $uv$, there is nothing to consider. However, if that is not the case, then, among the two paths from $u$ to $v$ along the boundary of $P$, consider the path $\partial P(u,v)$ which has a part visible to $uv$ (see Figure~\ref{fig:weight1}a); that is, there exists a line segment connecting the interior of $\partial P(u,v)$ to the interior of $uv$ that does not intersect the interior of $P$. 

$\partial P(u,v)$ concatenated with $uv$ then forms a closed curve $C$ on the plane that does not intersect the interior of $P$. $C$ is further subdivided by unsettled edges (non-crossing by planarity), with endpoints between vertices of $\partial P(u,v)$, into cells $c_1,\dots,c_{k}$ (see Figure~\ref{fig:weight1}b). 

\begin{figure}
\begin{center}
    \includegraphics{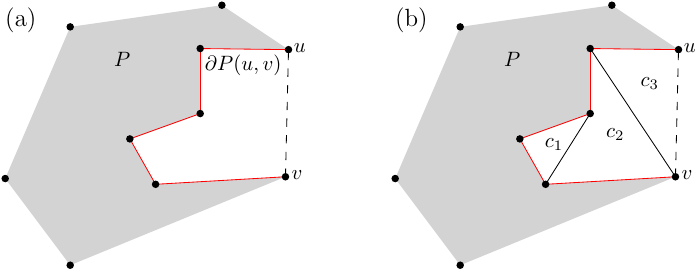}
    \caption{(a) $\partial P(u,v)$ has a part visible to $uv$. The dotted edge is a convex hull edge of $V$. (b) $\partial P(u,v) \cup uv$ is subdivided into $k$ cells. Each gray edge is an edge in $\mathcal{MBDG}(V)$.}
    \label{fig:weight1}
\end{center}
\end{figure}

If there are no unsettled edges, we expand $\partial P(u,v)$ into $uv$ by removing $\partial P(u,v)$ from $P$ and adding $uv$ to $P$. 
If, on the other hand, there is at least one unsettled edge, 
there must be some cell $c_i$ whose entire boundary, minus one unsettled edge $pq$, coincides with a part 
of $\partial P(u,v)$ (see Figure~\ref{fig:weight34}(a)). 
\begin{figure}[ht]
\begin{center}
    \includegraphics{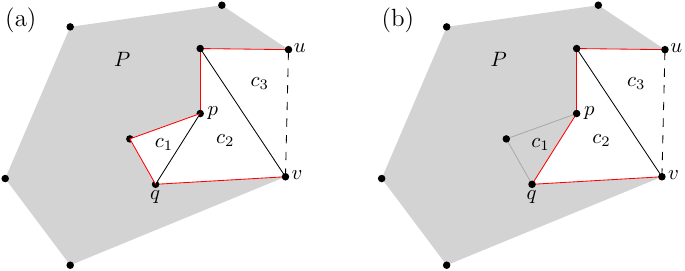}
\end{center} 
\caption{(a) $c_1$ coincides with part of $\partial P(u,v)$ except its one unsettled edge. $c_2$ and $c_3$ are not candidates for expansion. (b) Expansion of $P$ into $c_1$.}\label{fig:weight34}
\end{figure}
Then, we consider the addition of $pq$ into $\mathcal{LMBDG}(V)$, make it a settled edge, and expand $P$ into $c_i$ by removing the subpath from $p$ to $q$ along $\partial P(u,v)$ from $P$, and adding the edge $pq$ to $P$ (see Figure~\ref{fig:weight34}(b)). 
Since the area of $P$ is increasing, this process must eventually terminate. 
\subsubsection{Condition for Including an Edge}
The decision whether to include an edge in $\mathcal{LMBDG}(V)$ depends on an adjustable parameter $r > 0$, 
which causes an increase in the stretch factor by a factor of at most $1 + 1/r$ and ensures a weight of at most $(2r + 1)$ times that of $MST(\mathcal{MBDG}(V))$. 

All settled edges are assigned a $weight \geq 0$, which is the length of a short (but not necessarily shortest) path between their 
endpoints that uses only the currently included settled edges, which are by definition edges of $\mathcal{LMBDG}(V)$. 
Initially, $weight(pq) = \Abs{pq}$ for all edges $pq$ in the minimum spanning tree of $\mathcal{MBDG}(V)$. 
Now, when considering whether to include the unsettled edge $uv$ into $\mathcal{LMBDG}(V)$, we take the sum $S$ of the $weight$ of edges in $\partial P(u,v)$. 
These edges have been settled and thus have $weight$ assigned. 
If $S$ is greater than $(1 + 1/r)\cdot |uv|$, add $uv$ to $\mathcal{LMBDG}(V)$ and assign it a $weight$ of $|uv|$ now that is has been settled. 
Otherwise, settle $uv$ but exclude it, and assign it a $weight$ of $S$. We can see that $S$ is the length of the path from $u$ to $v$; that is, the concatenation of paths between the endpoints of edges in $\partial P(u,v)$. 

\subsection{Bounds on the Levcopoulos and Lingas Protocol}
Given an unsettled edge $uv$, let $\partial P(u,v)$ be the path along $P$ from $u$ to $v$ such that $\partial P(u,v)$ concatenated with $uv$ forms a closed curve 
that does not intersect the interior of $P$. When processing an edge $uv$, it is added to $\mathcal{LMBDG}(V)$ when the summed weight of the edges of $\partial P(u,v)$ is greater than $(1 + 1/r) \cdot |uv|$. This implies that $\mathcal{LMBDG}(V)$ is a spanner. 

\begin{theorem}\label{thm:weight1}
    $\mathcal{LMBDG}(V)$ is a $(1 + 1/r)$-spanner of $\mathcal{MBDG}(V)$ for an adjustable parameter $r > 0$. 
\end{theorem}

\begin{theorem}\label{thm:weight2}
    $\mathcal{LMBDG}(V)$ has weight at most $(2r + 1)$ times the weight of the minimum spanning tree of $\mathcal{MBDG}(V)$ for an adjustable parameter $r > 0$. 
\end{theorem}
\begin{proof}
    Let $P$ be the polygon that encloses $V$ in the above algorithm. Initially $P$ is the degenerate polygon described by the Euler tour of the minimum spanning tree of $V$ in $\mathcal{LMBDG}(V)$.
    Give each edge $e$ of $P$, a starting credit of $r|e|$. Denote the sum of credits of edges in $P$ with $credit(P)$. The sum of $credit(P)$ and the weight of the initially included settled edges is then $(2r + 1)$ times the weight of the minimum spanning tree of $\mathcal{MBDG}(V)$. 

    As $P$ is expanded and edges are settled, we adjust the credits in the following manner:
    \begin{itemize}
        \item If an edge $uv$ is added into $\mathcal{LMBDG}(V)$ when settled, we set the credit of the newly added edge $uv$ of $P$ to $credit(\partial P(u,v)) - |uv|$, and the credit of the edges along $\partial P(u,v)$ to 0. 
        \item If an edge is excluded from $\mathcal{LMBDG}(V)$ when settled, we set the credit of the newly excluded edge $uv$ of $P$ to $credit(\partial P(u,v))$, and the credit of edges along $\partial P(u,v)$ to 0. 
    \end{itemize}
    We can see that the sum of $credit(P)$ and the weights of included settled edges, at any time, is at most $2r + 1$ times the weight of the minimum spanning tree of $\mathcal{MBDG}(V)$ since it strictly drops when adding an edge when it is settled and stays the same when excluding an edge when it is settled. 

    It now suffices to show that $credit(P)$ is never negative, which we do by showing that for every edge $uv$ of $P$, at any time, $credit(uv) \geq r\cdot weight(uv) \geq 0$. We do this by induction over the edges in the order they are settled. 
    For the base case, when $P$ is the Euler Tour around the minimum spanning tree of $\mathcal{MBDG}(V)$, we have that $credit(uv) = r\cdot weight(uv)$. 
    For the induction step, let $uv$ be a settled edge. There are two cases:

    (I) If $uv$ is added to $\mathcal{LMBDG}(V)$, then $credit(uv)$ equals 
    \begin{align*}
        credit(\partial P(u,v)) - \Abs{uv} \geq\; & r \cdot weight(\partial P(u,v)) - \Abs{uv}\\
        \geq\; & r (1 + 1/r) \Abs{uv} - \Abs{uv}\\
        =\; & r\cdot weight(uv).
    \end{align*}
    The first inequality holds from the induction hypothesis, and the second inequality and last equality hold since $uv$ is added to $\mathcal{LMBDG}(V)$. 

    (II) If $uv$ is not added to $\mathcal{LMBDG}(V)$, then $credit(uv)$ equals
    \[credit(\partial P(u,v)) \geq r \cdot weight(\partial P(u,v)) = r\cdot weight(uv).\]
    The first inequality holds from the induction hypothesis, and the equality holds since $uv$ was not added. 

    Since $credit(P)$ is never negative, and the sum of $credit(P)$ and the weights of included settled edges is at most $2r + 1$ times the weight of the minimum spanning tree of $\mathcal{MBDG}(V)$, 
    the theorem follows. 
\end{proof}

\noindent
Putting together all the results so far, we get:

\begin{theorem} \label{thm:LMBDG}
    Given a set $V$ of $n$ points in the plane together with two parameters $0 < \theta < \pi/2$ and $r > 0$, one can compute in $O(n \log n)$ time a planar graph $\mathcal{LMBDG}(V)$ that has degree at most $5\Ceil{2\pi/\theta}$, weight of at most $((2r+1)\cdot \tau)$ times that of a minimum spanning tree of $V$, and is a $((1+1/r) \cdot \tau)$-spanner of $V$, where $\tau=1.998 \cdot \max(\pi/2,\pi \sin(\theta/2) + 1)$. 
\end{theorem}
\begin{proof}
    Let us start with the running time. The algorithm by Levcopoulos and Lingas (Lemma~3.3 in~\cite{levcopoulos1992there}) can be implemented in linear time and, according to Corollary~\ref{cor:desideratum1}, $\mathcal{BDG}(V)$ can be constructed in $O(n \log n)$ time, hence, $O(n \log n)$ in total.
    
    The degree bound and planarity follow immediately from the fact that $\mathcal{LMBDG}(V)$ is a subgraph of $\mathcal{MBDG}(V)$, and the bound on the stretch factor follows from Theorem~\ref{thm:weight1} and Corollary~\ref{cor:desideratum1}. 
    
    It only remains to bound the weight. Callahan and Kosaraju~\cite{ck-fasgg-93} showed that the weight of a minimum spanning tree of a Euclidean graph $G(V)$ is at most $t$ times that of the weight of $MST(V)$ whenever $G$ is a $t$-spanner on $V$. Since $\mathcal{MBDG}(V)$ is a $\tau$-spanner on $V$ by Corollary~\ref{cor:desideratum1}, $\mathcal{LMBDG}(V)$ has weight of at most $((2r+1)\cdot \tau)$ times that of the minimum spanning tree of $V$ by Theorem~\ref{thm:weight2}. This concludes the proof of the theorem. 
\end{proof}

\noindent
Finally, we prove that $\mathcal{LMBDG}(V)$ has short paths between the ends of pruned edges. 

\begin{theorem}
    Let $uv$ be an excluded settled edge. There is a face path in $\mathcal{LMBDG}(V)$ from $u$ to $v$ of length at most $(1 + 1/r)\cdot |uv|$. 
\end{theorem}
\begin{proof}
    If $uv$ is the first excluded settled edge processed by the Levcopoulos-Lingas algorithm, then all edges of $\partial P(u,v)$ must be included in $\mathcal{LMBDG}(V)$. By planarity, no edge will be added into the interior of the cycle consisting of $uv$ and $\partial P(u,v)$ once $uv$ is settled, and thus $uv$ will be a chord on the face in $\mathcal{LMBDG}(V)$ that coincides with $\partial P(u,v)$.  Thus, $\partial P(u,v)$ is a face path in $\mathcal{LMBDG}(V)$ from $u$ to $v$ with a length of at most $weight(uv) \leq (1+1/r)\cdot |uv|$. 

    Otherwise, if $uv$ is an arbitrary excluded edge, then some edges of $\partial P(u,v)$ may be excluded settled edges. If none are excluded, then $\partial P(u,v)$ is again a face path with length at most $weight(uv)$. However, if some edges are excluded, then, by induction, for each excluded edge $pq$ along $\partial P(u,v)$, there is a face path in $\mathcal{LMBDG}(V)$ from $p$ to $q$ with a length of $weight(pq) \leq (1 + 1/r)\cdot |pq|$. 
    Replacing all such $pq$ in $\partial P(u,v)$ by their face paths, and since no edge will be added into the interior of the cycle consisting of $uv$ and $\partial P(u,v)$ once $uv$ is settled, $\partial P(u,v)$ with its excluded edges replaced by their face paths is a face path in $\mathcal{LMBDG}(V)$ from $u$ to $v$ with a length of $weight(uv) \leq (1 + 1/r)\cdot |uv|$. 
\end{proof}

\noindent {\bf Remark:} 
The remainder of this section is not required to proceed, but it is worth mentioning this curious phenomenon. We can say something even stronger about the weight if $\theta$ is small. 
When $\theta \leq \pi/3$, 
the weight of $\mathcal{LMBDG}(V)$ can be bounded to be no more than $(2r+1)$ times that of a minimum spanning tree on $V$.

\begin{lemma}\label{lem:tree1}
    If $uv_1$ and $uv_2$ are edges in a minimum spanning tree of $V$ then $\angle(v_1,u,v_2) \geq \pi/3$. 
\end{lemma}
\begin{proof}
    Refer to Figure~\ref{fig:trees}(a). 
    Let $uv_1$ and $uv_2$ be edges in a minimum spanning tree of $V$. 
    Suppose for a contradiction that $\angle(v_1,u,v_2) < \pi/3$. 
    Then, without loss of generality, we can say that $\angle(u,v_1,v_2) > \pi/3$. 
    Since $\angle(v_1,u,v_2) < \pi/3$ and $\angle(u,v_1,v_2) > \pi/3$, we deduce that $\Abs{v_1v_2} < \Abs{uv_2}$. 
    We can therefore replace $uv_2$ with $v_1v_2$ to get a lighter spanning tree, contradicting the minimality of the tree. 
    Therefore, it must be that $\angle(v_1,u,v_2) \geq \pi/3$. 
\end{proof}
\begin{lemma}\label{lem:tree2}
    Fix a minimum spanning tree on $V$. 
    Let $C$ be a cone with apex $u$ and angle measure less than $\pi/3$. 
    If $uv$ is a minimum spanning tree edge contained in $C$, and if there is a $w \in V \cap C$ such that $\Abs{uw} \leq \Abs{uv}$, then we can replace $uv$ with $uw$ to get another minimum spanning tree. 
\end{lemma}
\begin{proof}
    Fix a minimum spanning tree on $V$. 
    Let $C$ be a cone with apex $u$ and angle measure less than $\pi/3$, and let $uv$ be a minimum spanning tree edge contained in $C$. 
    Suppose there is a $w \in V \cap C$ such that $\Abs{uw} \leq \Abs{uv}$. 
    We consider two cases separately. 
    In the first case, when the path in the minimum spanning tree from $v$ to $w$ does not go through $u$ (see Figure~\ref{fig:trees}(b)), we can replace $uv$ with $uw$ to get a spanning tree no heavier. 
    In the second case, when the path in the minimum spanning tree from $v$ to $w$ goes through $u$ (see Figure~\ref{fig:trees}(c)), we can replace $uv$ with $vw$ to get a lighter spanning tree. 
    This is a contradiction to the minimality of the spanning tree and is thus an impossible case. 
    Since we can fix up the first case, and the second is impossible, we have shown how that we can replace $uv$ with $uw$ to get another minimum spanning tree. 
\end{proof}

    \begin{figure}[ht]
    \begin{center}
       \includegraphics[width=\textwidth]{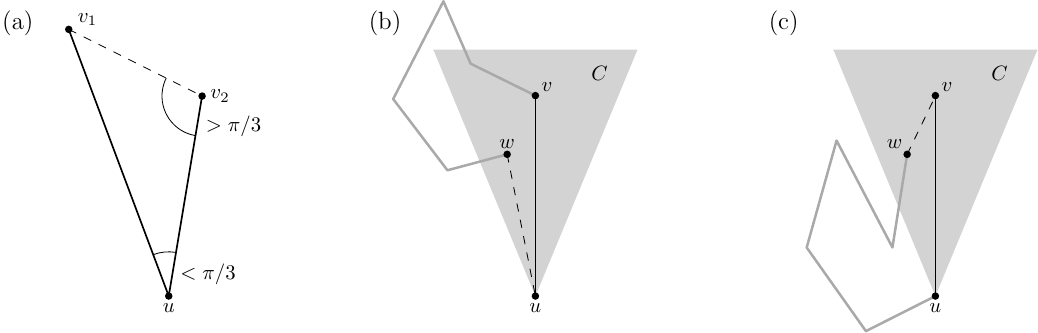}
    \end{center}
    \caption{(a) We can replace $uv_2$ with $v_1v_2$ to get a lighter tree. (b) We can replace $uv$ with $uw$ to get a tree no heavier. (c) We can replace $uv$ with $vw$ to get a lighter tree.}\label{fig:trees}
    \end{figure}

\begin{theorem}\label{thm:tree3}
    If $\theta < \pi/3$, a minimum spanning tree of $V$ is contained in $\mathcal{MBDG}(V)$. 
\end{theorem}
\begin{proof}
    It is a known fact that the Delaunay triangulation $\mathcal{DT}(V)$ contains a minimum spanning tree of $V$ (see~\cite{berg2008computational}). 
    Fix a minimum spanning tree of $\mathcal{DT}(V)$. 
    Suppose $uv$ is a minimum spanning tree edge that is not in $\mathcal{MBDG}(V)$. 
    It is therefore not fully protected. 
    Without loss of generality, say it is not protected at $u$. 
    Look at the cone $C \in \mathcal{C}_{u,\kappa}$ that contains $uv$.
    Let $um$ be the middle edge in $C$.
    Since the angle measure of $C$ is less than $\pi/3$, and 
    by the definition of the middle edge which says $\Abs{um} \leq \Abs{uv}$, we can replace $uv$ with $um$ by Lemma~\ref{lem:tree2} to get another minimum spanning tree. 
    Since the angle measure of $C$ is less than $\pi/3$, there can only be one such edge $uv$ in $C$ that needs replacement, by Lemma~\ref{lem:tree1}. 
    This says that we are replacing at most one minimum spanning tree edge with $um$, to get another minimum spanning tree. 
    Repeating this process for all minimum spanning tree edges that are not fully protected, we will trade a set of $k\geq0$ distinct minimum spanning tree edges that are not fully protected with $k$ distinct middle edges to get another minimum spanning tree; one that is contained in $\mathcal{MBDG}(V)$. 
\end{proof}
\begin{corollary}\label{cor:thresh}
    If $\theta < \pi/3$, $\mathcal{LMBDG}(V)$ has a weight no more than $2r+1$ times that of a minimum spanning tree of $V$. 
\end{corollary}
\begin{proof}
    This follows immediately from Theorem~\ref{thm:tree3}. 
\end{proof}

\section{Routing on the Light Graph}
\label{sec:routing2}
In order to route on $\mathcal{LMBDG}(V)$, we store edge-data at each of its endpoints when it is excluded. Specifically, let $uv$ be some excluded edge; at $u$ (and $v$) we store $uv$, along with one bit to indicate whether the starting edge of the $(1+1/r)$-path is the edge clockwise or counterclockwise to $uv$. 

\begin{observation}
    $\mathcal{LMBDG}(V)$ stores $O(1)$ words of information at each vertex. 
\end{observation}

To route on $\mathcal{LMBDG}(V)$, we simulate the routing algorithm on $\mathcal{MBDG}(V)$. When this algorithm would follow an excluded edge $uv$ at $u$, we store $v$ and the orientation of the face path from $uv$ at $u$ in memory. 
Then, until $v$ is reached, take the edge that is clockwise or counterclockwise to the edge arrived from, in accordance with the orientation stored. Once $v$ is reached, we proceed with the next step of the routing algorithm on $\mathcal{MBDG}(V)$.

Note that bounding the weight in this manner only requires the input graph to be planar. It transforms the pruned edges into $O(d)$ information at each vertex, where $d$ is the degree of the input graph; in our case $d$ is a constant. The scheme of simulating a particular routing algorithm and switching to a face routing mode when needed can then be applied to the resulting graph. 

\begin{theorem}\label{thm:weightasdasd}
    The routing algorithm on $\mathcal{LMBDG}(V)$ is $1$-local, has a routing ratio of $5.90(1 + 1/r)\max(\pi/2 , \pi \sin(\theta / 2) + 1)$ and uses $O(1)$ memory. 
\end{theorem}
\begin{proof}
    The $1$-locality follows by construction. 
    The routing ratio follows from Theorem~\ref{thm:route10}.
    Finally, the memory bound follows from the fact that while routing along a face path to get across a pruned edge, no such subpaths can be encountered. Thus, the only additional memory needed at any point in time is a constant amount to navigate a single face path. 
\end{proof}

\section{Conclusion}
We showed how to construct and route locally on a bounded-degree lightweight spanner. In order to do this, we simulate a relaxation of Bonichon~\etal's routing algorithm~\cite{bonichon2017upper} on Delaunay triangulations. A natural question is whether our routing algorithm can be improved by using the improved Delaunay routing algorithm by Bonichon~\etal~\cite{bonichon2018improved}. Unfortunately, this is not obvious: when applying the improved algorithm on our graph, we noticed that the algorithm can revisit vertices. While this may not be a problem, it implies that the routing ratio proof from~\cite{bonichon2018improved} needs to be modified in a non-trivial way and thus we leave this as future work. 

\bibliography{References}

\newpage
\appendix

\section{Algorithmic Construction of $\boldsymbol{\mathcal{MBDG}(V)}$} \label{app:MBDG}
We state a modification to the construction of $\mathcal{BDG}(V)$ to construction of $\mathcal{MBDG}(V)$.
Added lines to the former construction have been inserted at lines \texttt{20,21,22} and coloured red.
\begin{algorithm}
    \caption{$\mathcal{MBDG}(V)$}
    \begin{algorithmic}[1]
        \Require $V$
        \Require $0 < \theta < \pi/2$
        \State $E \gets \{\}$
        \State $\mathcal{DT} \gets \mathcal{DT}(V)$
        \For{$u \in V$}
        \State Compute $\mathcal{C}_{u,\kappa}$, where $\kappa = \Ceil{2\pi/\theta}$.
        \EndFor
        \For{$u \in V$}
            \For{$uv \in E(\mathcal{DT})$}
                \State Bucket $uv$ into $C \in \mathcal{C}_{u,\kappa}$.
            \EndFor
        \EndFor
        \For{$u \in V$}
            \For{$C \in \mathcal{C}_{u,\kappa}$}
                \State Reset values of $e_1, e_2, p_1, p_2, m$.
                \For{$e$ bucketed into $C$}
                    \State $e_1 \gets \text{arg}\min_{angle}(e_1, e)$
                    \State $e_2 \gets \text{arg}\max_{angle}(e_2, e)$
                \EndFor
                \For{$e$ bucketed into $C \backslash \{e_1,e_2\}$}
                    \State $p_1 \gets \text{arg}\min_{angle}(p_1, e)$
                    \State $p_2 \gets \text{arg}\max_{angle}(p_2, e)$
                \EndFor
                \For{$e$ bucketed into $C \backslash \{e_1,e_2,p_1,p_2\}$}
                    \State $m \gets \text{arg}\min_{length}(m, e)$
                \EndFor
                \State Mark $e_1,e_2, p_1,p_2,m$, if their values are set, as protected by $u$.
                {
                \color{red}
                \For{$uv$ bucketed into $C \backslash \{e_1,e_2,p_1,p_2,m\}$}
                    \State Store $uv$ at $v$ as a semi-protected edge, marked with $1$ if it's to the right of $m$, and $0$ otherwise.
                \EndFor
                \State Mark $e_1,e_2, p_1,p_2, m$ as extreme, penultimate, or middle at $u$, if their values are set.
                }
            \EndFor
        \EndFor
        \For{$uv \in E(\mathcal{DT})$}
            \If{$uv$ marked as protected by both endpoints}
                \State $E = E \cup \{uv\}$.
            \EndIf
        \EndFor
        \Return $(V,E)$.
    \end{algorithmic}
\end{algorithm}

\begin{theorem}\label{thm:construct}
    $\mathcal{MBDG}(V)$ takes $O(n\log n)$ time to construct.
    $\mathcal{MBDG}(V)$ takes $O(n)$ time to construct if the input is a Delaunay triangulation $\mathcal{DT}(V)$ on $V$.
\end{theorem}
\begin{proof}
    The running time of the loop at line \texttt{8} remains unchanged; it is $O(n)$ since there are a linear number of edges, each looked at at most eight times (four times per endpoint).
    For the same reasons that justify the construction time of $\mathcal{BDG}(V)$, we can then conclude that
    $\mathcal{MBDG}(V)$ takes $O(n\log n)$ time to construct and
    $\mathcal{MBDG}(V)$ takes $O(n)$ time to construct if the input is a Delaunay triangulation $\mathcal{DT}(V)$ on $V$. 
\end{proof}

\end{document}